\documentclass{article}
\usepackage{graphicx}
\usepackage{amsmath, amssymb,amsthm}
\usepackage{subcaption}
\usepackage{float}
\usepackage{booktabs}
\usepackage{array}
\usepackage{xcolor}
\usepackage{esint}
\usepackage{enumitem}
\usepackage{listings}
\usepackage[cachedir=minted-cache]{minted}
\usepackage[toc,page]{appendix}
\usepackage{lipsum} 
\usepackage{url}
\usepackage[a4paper, margin=0.75in]{geometry}
\usepackage[colorlinks=true, linkcolor=blue, citecolor=blue, urlcolor=blue]{hyperref}
\newcommand{\ii}{\mathrm{i}}
\usepackage{graphicx}
\usepackage{subcaption} 

\theoremstyle{definition}
\newtheorem{definition}{Definition}[section]

\newtheorem{theorem}{Theorem}
\usepackage{algorithm}
\usepackage{algpseudocode}
\usepackage[font=small]{subcaption}
\usepackage[font=small,labelfont=bf]{caption}
\newcommand{\vect}[1]{\mathbf{#1}}

\usepackage{authblk}

\title{Density Functions and Random Number Generators of $\alpha$-Stable Distributions}
\author[1]{Wael Tabbara}
\author[2]{Sharafeddine Sharafeddine\thanks{The majority of this author's contributions were completed while affiliated with the American University of Beirut.}}
\author[1]{Ahmad El-Hajj}
\author[1]{Jihad Fahs}
\author[1]{Ibrahim Abou-Faycal}
\affil[1]{Electrical \& Computer Eng. Department, American University of Beirut, Lebanon}
\affil[2]{School of Electrical \& Computer Engineering, Cornell University, USA}
\date{}

\begin{document}

\maketitle


\begin{abstract}
Heavy-tailed distributions are increasingly found to better fit empirical data in engineering, finance, physics, network science, and related fields. Among them, $\alpha$-stable distributions play a central role being limiting laws in the generalized central limit theorem: they are expected to be exceptionally good models whenever sums of multiple independent heavy-tailed sources are at play. Despite their theoretical importance, their practical use remains challenging: $\alpha$-stable probability densities generally do not have closed-form expressions, and numerical evaluation and random variate generation can be difficult, especially in the multivariate setting. This paper presents AUB-HTP, a Python package for numerical computation and simulation of $\alpha$-stable distributions. The package provides scalar density evaluation using several complementary methods, including Zolotarev-type integral representations, series formulas, and numerical inversion of characteristic functions. It also provides random variate generation for scalar and multivariate $\alpha$-stable distributions, with support for flexible spectral measures through LePage series representations. Numerical experiments demonstrate that AUB-HTP improves the accuracy, stability, and parameter coverage of existing tools for scalar density computation, while adding new capabilities for multivariate simulation. The package is designed to support reproducible computational work involving heavy-tailed models across a broad range of scientific applications.
\end{abstract}

\section{Introduction}

Heavy-tailed distributions arise naturally in many scientific and engineering applications, including communications, signal processing, finance, physics, and network modeling~\cite{marcel2008,Tsih1995,crovella1998heavy,nolancirc2013,taleb2022statisticalconsequencesfattails}. In such settings, rare but large deviations occur more frequently than predicted by light-tailed models such as the Gaussian distribution. As a result, Gaussian-based modeling can significantly underestimate the probability and impact of extreme events.

Among heavy-tailed models, $\alpha$-stable distributions occupy a central theoretical position. They generalize the Gaussian distribution through the stability property: sums of independent $\alpha$-stable random variables remain $\alpha$-stable, up to location and scale changes. They also arise as limiting distributions for normalized sums of independent random variables under the generalized central limit theorem~\cite{kolmo1968}. These properties make $\alpha$-stable laws natural candidates for modeling aggregate effects generated by many independent heavy-tailed sources.

A major practical limitation of $\alpha$-stable distributions is that their probability density functions (PDF)s generally do not admit closed-form expressions, except in the special cases of the Cauchy and Lévy distributions. Consequently, numerical tools are required for density evaluation, modeling, and simulation. Existing approaches for scalar $\alpha$-stable random variate generation, such as the Chambers--Mallows--Stuck (CMS) method~\cite{ChambersJ.M.1976AMfS}, are well established, but they do not directly address general multivariate $\alpha$-stable distributions. In the multivariate case, available simulation methods typically focus on special spectral measures, such as isotropic or sub-Gaussian models~\cite{nolancirc2013}, rather than arbitrary spectral measures. 

Available numerical packages for $\alpha$-stable distributions include \cite{Royuela-del-ValJavier2017l:FP} that provides an R and C++ implementation for $\alpha$-stable random variable generation using a modified version of the CMS method~\cite{ChambersJ.M.1976AMfS,NolanJohnP2020USDM}. In \cite{teimouri2018alphastablerpackagemodelling} the authors provide an R package that supports generating an $\alpha$-stable random variable using the same modified CMS method; a random variable from a truncated $\alpha$-stable distribution distribution; an $\alpha$-stable random vector with a sub-Gaussian spectral measure, and other bivariate $\alpha$-stable random vectors. The work in \cite{RJ-2022-056} presents another $\alpha$-stable package which references \cite{teimouri2018alphastablerpackagemodelling} as well. When it comes to density computations, the \texttt{libstable} package \cite{Royuela-del-ValJavier2017l:FP} is based on Zolotarev's integral representation \cite{Zolotarev1986} computed using Gauss-Kronrod quadrature over the entire parameter space. The package does the evaluations in C/C++ and provides a MATLAB front-end in addition to an R package. The \texttt{alphastable} R package \cite{teimouri2018alphastablerpackagemodelling} evaluates the density using Skorohod-type asymptotic series \cite{sko2,Pollard1946LaplaceIntegral,Bergstrom1952ExpansionsStable} wherever those series converge. The package falls back to \texttt{stabledist} \cite{Wuertz2016Stabledist} elsewhere, which computes the density via Zolotarev's integral representation \cite{Zolotarev1986} using adaptive quadrature.

This paper presents AUB-HTP, a Python package for computational work with $\alpha$-stable distributions. The package provides scalar density computation using several numerical methods, including Zolotarev integral representations~\cite{Zolotarev1986}, Skorohod-Pollard-Bergstr{\"o}m formulas~\cite{sko2,Pollard1946LaplaceIntegral,Bergstrom1952ExpansionsStable}, and direct numerical inversion of the \textit{characteristic function}\footnote{The characteristic function of a random variable $X$ is given by the expected value of $e^{\ii tX}$, 
\(
\varphi_X(u) = \mathbb{E}[e^{\ii uX}].
\)
It uniquely determines the distribution of $X$, and it always exists.}. These methods provide complementary tradeoffs in terms of accuracy, speed, and robustness across parameter regimes. In addition, AUB-HTP implements random variate generation for scalar and multivariate $\alpha$-stable distributions. For the multivariate case, the package uses LePage series representations~\cite{LePageRaoul1981CtaS}, which support flexible spectral measures and perform particularly well in challenging heavy-tailed regimes, including small values of $\alpha$.

The goal of AUB-HTP is to make $\alpha$-stable modeling more accessible to computational scientists by providing reliable density evaluation, flexible simulation tools, and reproducible numerical routines within the Python ecosystem. The main contributions of this work are: (i) a unified Python implementation of several scalar $\alpha$-stable PDF evaluation methods; (ii) a flexible random variate generator for scalar and multivariate $\alpha$-stable distributions; (iii) support for general spectral measures through LePage series; and (iv) numerical benchmarks comparing accuracy, runtime, and robustness against existing implementations.

\paragraph*{Notations}
\label{sect : notations}
We adopt the following conventions. The vector $\mathbf{x}\in\mathbb{R}^d$ represents a column vector and $\mathbf{x}^T$ is its transpose. The Euclidean inner product of two vectors  $\mathbf{x} = (x_1,x_2,\cdot\cdot\cdot,x_d)^T$ and $\mathbf{y} = (y_1,y_2,\cdot\cdot\cdot,y_d)^T$ is denoted by $\mathbf{x}^T\mathbf{y}=\langle\mathbf{x},\mathbf{y}\rangle=\sum_{i=1}^{d}x_iy_i $ and the $L_2$ norm of $\mathbf{x}$ by $\|\mathbf{x}\|=\sqrt{\langle\mathbf{x},\mathbf{x}\rangle}$. The unit sphere in $\mathbb{R}^d$ is denoted by $\mathbb{S}^{d-1}=\{\mathbf{x}\in\mathbb{R}^d: \|\mathbf{x}\|=1\}$. The class of all borel subsets in a metric space $S$ is represented by $\mathcal{B}(S)$. 
We use $\mathbf{X} \perp\!\!\!\perp \mathbf{Y}$ to indicate that the two random vectors $\mathbf{X}$ and $\mathbf{Y}$ are independent. 
When a random vector $\mathbf{X}$ is distributed uniformly on the unit sphere, we write $\mathbf{X}\sim\mathcal{U}(\mathbb{S}^{d-1})$. For a random variable $X$ that is Gamma distributed, we write $X\sim\text{Gamma}(\alpha,\beta)$, and we use the following parameterization for the PDF \begin{equation*}
    f_{X}(x)=\begin{cases}
        \displaystyle \frac{1}{\Gamma(\alpha)\beta^{\alpha}} \, x^{\alpha-1} e^{-\frac{x}{\beta}} & \text{if}\; x>0\\
        0 & \text{otherwise},
    \end{cases}
\end{equation*} 
where $\Gamma(\cdot)$ is the Gamma function. 


\section{Preliminaries on $\alpha$-Stable Distributions}
For a definition of multivariate stable vectors, we refer the reader to \cite[Definition 2.1.1]{SamorodnitskyGennady2000Snrp}.
In particular, we are interested in $\alpha$-stable~\footnote{The term $\alpha$-stable is used to refer to the stable family excluding the Gaussian law.} random vectors. 
 
 \begin{definition}~\cite{SamorodnitskyGennady2000Snrp}
 \label{def:astable}
 Let $0<\alpha<2$, then $\mathbf{X}=(X_1,X_2,\cdot\cdot\cdot,X_d)$ is an $\alpha$-stable random vector in $\mathbb{R}^d$ if and only if there exists a finite measure \(\Lambda\) on $\mathbb{S}^{d-1}$ called the ``spectral" measure and a vector \(\boldsymbol{\mu}^{0}\) in $\mathbb{R}^d$ such that the characteristic function of $\mathbf{X}$ is given by 
 \begin{equation}\label{eq:eq3}
    \phi_{\mathbf{X}}(\mathbf{t}) = \mathbb{E} \left[ \exp\{\ii \langle \mathbf{t},\mathbf{X}\rangle\} \right]= \exp\left\{-\int_{\mathbb{S}^{d-1}} \psi_\alpha(\langle \mathbf{t}, \mathbf{s} \rangle) \, \Lambda(d\mathbf{s}) + \ii \langle \mathbf{t},\boldsymbol{\mu}^0 \rangle\right\},
  \end{equation}
    and where \begin{equation*}
    \psi_\alpha(u) = 
\begin{cases}
\displaystyle |u|^\alpha \left[ 1 - \ii \operatorname{sign}(u) \tan \frac{\pi \alpha}{2} \right] & \alpha \ne 1 \\
\displaystyle \, |u| \left[ 1 + \ii \frac{2}{\pi} \operatorname{sign}(u) \ln |u| \right] & \alpha = 1.
\end{cases}
\end{equation*}
\end{definition}

In the scalar case ($d = 1$) --and whenever $0<\alpha<2$, the spectral measure boils down to two mass points on $s = \pm 1$, and the characteristic function in~\eqref{eq:eq3} is expressible in terms of 4 parameters: the stability index $0<\alpha<2$, skewness $-1\leq\beta\leq1$, scale $\gamma>0$, and location $\delta\in\mathbb{R}$. In this case, we denote an \(\alpha\)-stable random variable $X$ using the same notation as in Nolan~\cite{NolanJohnP2020USDM}, \(X\!\sim\!\mathcal{S}(\alpha,\beta,\gamma,\delta;k)\) where $k$ is the parameterization index (or type). We will use the symbols $\mathcal{S}_0$ and $\mathcal{S}_1$ to refer to the two common type-0 and type-1 parameterization, respectively. In the multivariate case, we denote a $d$-dimension non-degenerate $\alpha$-stable random vector \(\vect{X}\) with stability index $\alpha$, spectral measure $\Lambda$ and shift vector $\boldsymbol{\mu}^0$, by $\vect{X}\sim\mathbf{S}_\alpha(\Lambda,d,\boldsymbol{\mu}^0)$.
Note that Definition~\ref{def:astable} corresponds to the generalization of the scalar case under the parameterization $k=1$.

\subsection*{Parameterizations of scalar $\alpha$-stable distributions}

The two parameterizations, 0 and 1 are denoted by $\mathcal{S}(\alpha, \beta, \gamma, \delta; 0)$ and $\mathcal{S}(\alpha, \beta, \gamma, \delta; 1)$ respectively and differ primarily in how they define the location (shift) of the distribution, particularly near $\alpha = 1$. For a given $\alpha$ and $\beta$, let $Z$ be a ``standard" $\alpha$-stable variable corresponding to $\gamma = 1$ and $\delta = 0$.

\begin{table}[!htb]
\begin{center}
\renewcommand{\arraystretch}{1.5}

\begin{tabular}{|l|l|l|} \hline
$X \sim \mathcal{S}(\alpha, \beta, \gamma, \delta; 0)$ & Characteristic Function: $\mathbb{E}[e^{iuX}]$ & $\alpha$ \\ \hline \hline
$X \overset{d}{=} \gamma \left( Z - \beta \tan\left( \frac{\pi \alpha}{2} \right) \right) + \delta$ & $\exp\left( -|\gamma u|^\alpha \left[ 1 + \ii\beta \tan\left( \frac{\pi \alpha}{2} \right) \operatorname{sign}(u)(|\gamma u|^{1 - \alpha} - 1) \right] + \ii \delta u \right)$ & $\alpha \neq 1$ 
\\ 
$X \overset{d}{=} \gamma Z + \delta$ & $\exp\left( -|\gamma u| \left[ 1 + \ii\beta \frac{2}{\pi} \operatorname{sign}(u) \log|\gamma u| \right] + \ii \delta u \right)$ & $\alpha = 1$
\\ \hline
 \multicolumn{3}{c}{} \\ \hline
$X \sim \mathcal{S}(\alpha, \beta, \gamma, \delta; 1)$ & Characteristic Function: $\mathbb{E}[e^{iuX}]$ & $\alpha$ \\ \hline \hline
$X \overset{d}{=} \gamma Z + \delta$ & $\exp\left( -|\gamma u|^\alpha \left[ 1 - \ii\beta \tan\left( \frac{\pi \alpha}{2} \right) \operatorname{sign}(u) \right] + \ii \delta u \right)$ & $\alpha \neq 1$ 
\\ 
$X \overset{d}{=} \gamma Z + \left( \delta + \beta \frac{2}{\pi} \gamma \log \gamma \right)$ & $\exp\left( -|\gamma u| \left[ 1 + \ii\beta \frac{2}{\pi} \operatorname{sign}(u) \log|u| \right] + \ii \delta u \right)$ & $\alpha = 1$
\\ \hline
\end{tabular}
\caption{Parametrization 0 and Parametrization 1. \label{eq:Type_!}}
\renewcommand{\arraystretch}{1}
\end{center}
\end{table}

Parameterization 0 has the advantage of avoiding discontinuities near $\alpha = 1$ by adjusting the shift behavior and is often preferred for numerical stability~\cite{NolanJohnP2020USDM}.
On the other hand, explicit adjustment to the location is done when $\alpha = 1$ for parameterization 1.

Note that the choice of parameterization directly affects the shift and symmetry of the distribution, especially when $\beta \neq 0$ and $\alpha$ is close to $1$. In both forms, the distribution can be standardized by setting $\gamma = 1$ and $\delta = 0$, in which case we use the abbreviations $\mathcal{S}(\alpha, \beta; 0)$ and $\mathcal{S}(\alpha, \beta; 1)$ respectively. 

In the univariate case, our package can simulate both parameterizations $k=0,1$.

\section{Scalar Density Evaluation}

The AUB-HTP package provides numerical evaluation of scalar $\alpha$-stable probability density functions through a SciPy-like interface. This component complements the random number generator: the generator produces samples, while the density module provides the reference curve used for model verification, visualization, and numerical comparison.

Except for  Cauchy and L\'evy laws, $\alpha$-stable distributions do not have a closed-form density, which must therefore be computed using numerical inversion or analytic representations that are valid in specific parameter regions.

\smallskip

Towards fulfilling our main objective of reducing expensive quadrature calls while preserving accuracy across a wide range of $\alpha$, $\beta$, and sample values, the density evaluator combines three methods:
\begin{itemize}
\item[(a)] direct inversion of Nolan's $\mathcal{S}_1$ characteristic function near the mode \cite{NolanJohnP2020USDM}
\item[(b)] Zolotarev's integral representation in the middle range \cite{Zolotarev1986}
\item[(c)] Skorohod-Pollard-Bergstr{\"o}m's series expansions in the tails \cite{sko2, Wintner1941CauchySingularities, Pollard1946LaplaceIntegral}.
\end{itemize}

\subsection{Input standardization}

All density calls are first mapped to a standardized variable allowing the internal routines to work on standardized densities while still supporting user-facing location and scale parameters. More specifically, given an input sample value $x$ the code evaluates
\begin{equation}
z=\frac{x-\text{shift}}{\gamma},
\label{eq:shift}
\end{equation}
then rescales the final density by $1/\gamma$. For the $\mathcal{S}_1$ parameterization for example, the shift is given by \cite[equation 1.5]{NolanJohnP2020USDM}
\[
\text{shift} =
\begin{cases}
\delta, & \alpha \neq 1\\[6pt]
\delta + \beta \dfrac{2}{\pi} \gamma \log \gamma & \alpha = 1,
\end{cases}    
\]
as can be deduced from Table~\ref{eq:Type_!}.

Additionally, if $f_{\alpha,\beta}(x)$ is the PDF of a normalized $\alpha$-stable random variable, for negative evaluation points the code uses the symmetry relation \cite[proposition 1.1]{NolanJohnP2020USDM}
\[
f_{\alpha,\beta}(-x)=f_{\alpha,-\beta}(x),
\]
and most formulas only need to be evaluated on the positive half-line.

\subsection{Hybrid evaluation strategy}

The main density function acts as a dispatcher. It divides the input array into regions and applies the method that is most accurate for each region.
\begin{itemize}
\item For small $|z|$ (see equation~\eqref{eq:shift}), the package uses direct inversion of the characteristic function. While this method is slower than the series expansions, it is stable near the mode where tail formulas are not accurate.
\item For intermediate values of $|z|$, the package uses Zolotarev's integral representation. This provides a useful compromise between direct inversion and asymptotic series methods.
\item For large $|z|$, the package uses Skorohod-Pollard-Bergstr{\"o}m series' expansions. These formulas are fast in the tails and avoid repeated numerical integration.
\end{itemize}

The switching points between methods are not fixed constants. They are instead selected from precomputed cutoff functions depending on $(\alpha,\beta)$. These cutoffs were obtained using grid-based comparison with \texttt{scipy.stats.levy\_stable.pdf}. At runtime, the code queries the corresponding cutoff and applies Boolean masks to route each input point to the selected method.

\subsubsection*{Density formulas}

$\bullet$ For direct inversion in the $\mathcal{S}_1$ parameterization~\cite{NolanJohnP2020USDM}, the density is computed from the characteristic function. The standardized density is
\begin{align*}
    & f_{\alpha,\beta}(x) = \frac{1}{\pi} \int_0^\infty e^{-t^\alpha} \cos\!\left( xt-\beta\tan\left(\frac{\pi\alpha}{2}\right)t^\alpha\right) \,dt 
    \quad \alpha \neq 1 \\
    & f_{1,\beta}(x) = \frac{1}{\pi} \int_0^\infty e^{-t} \cos\!\left( xt+\frac{2}{\pi}\beta t\log t \right) \,dt 
    \qquad \qquad \alpha = 1, \text{ with the logarithmic correction.}
\end{align*}
These integrals are evaluated with \texttt{scipy.integrate.quad\_vec}.

\smallskip
\noindent
$\bullet$ For the middle region, the code uses Zolotarev's integral representation~\cite{Zolotarev1986}. For $\alpha\neq1$ and $x>0$,
\[
f_{\alpha,\beta}(x) = \frac{\alpha x^{\frac{1}{\alpha-1}}} {\pi |\alpha-1|}
\int_{-\theta_0}^{\pi/2} V(\theta) \exp\!\left[ -x^{\frac{\alpha}{\alpha-1}}V(\theta)\right]\,d\theta,    
\]
where
\[
\theta_0= \frac{1}{\alpha}\arctan\!\left(\beta\tan\left(\frac{\pi\alpha}{2}\right)\right),
\]
and where $V(\theta)$ is evaluated following Zolotarev's formula~\cite{Zolotarev1986}. This integral is computed using the function \texttt{scipy.integrate.quad}. 

\smallskip
\noindent
$\bullet$ For the tail region, we use Skorohod-type series expansions \cite{sko2, Wintner1941CauchySingularities, Pollard1946LaplaceIntegral} of the form
\[
f_{\alpha,\beta}(x) \approx \frac{1}{\pi x} \sum_{n=1}^{N} a_n x^{-\alpha n},    
\]
where
\[
a_n = \frac{(-1)^{n-1}\Gamma(n\alpha+1)}{n!} \left( 1+\beta^2\tan^2\left(\frac{\pi\alpha}{2}\right) \right)^{n/2} \sin\!\left[n\left(
\frac{\pi\alpha}{2} + \arctan\!\left( \beta\tan\left(\frac{\pi\alpha}{2}\right) \right) \right) \right].    
\]
The number $N$ of terms used to calculate the series is appropriately chosen depending on the parameter range. In the current package, the $\alpha<1$ branch uses a larger truncation level, while the $\alpha>1$ branch requires fewer terms in the tested regimes.

In some cases, minor spikes occurred at disjoint points away from the mode due to numerical irregularities. These spikes were detected by checking locations where monotonicity is violated. This is resolved by applying --after the density is computed, a smoothing step enforcing monotonicity of $f_{\alpha,\beta}$ on either side of the mode.

\subsection{Validation}

The density evaluator was tested against \texttt{scipy.stats.levy\_stable.pdf} on grids of $\alpha$, $\beta$, and $x$. The near-mode tests used $x\in[-10,10]$, while the tail tests examined values up to $x = 100$.

The metrics listed in Table~\ref{tab:error-metrics} were used to compare the AUB-HTP PDF generator with that of SciPy.
\begin{table}[!htbp]
\centering
\renewcommand{\arraystretch}{2.2}
\begin{tabular}{@{}p{4cm} >{\centering\arraybackslash}m{7cm}@{}}
\toprule
\textbf{Metric} & \textbf{Formula} \\
\midrule
Maximum absolute error &
$\displaystyle \max_{1 \leq i \leq n}\left|\hat{f}(x_i)-f_{\mathrm{SciPy}}(x_i)\right|$ \\
Maximum relative error &
$\displaystyle \max_{1 \leq i \leq n}\frac{\left|\hat{f}(x_i)-f_{\mathrm{SciPy}}(x_i)\right|}{\max\left(\left|f_{\mathrm{SciPy}}(x_i)\right|,\varepsilon\right)}$ \\
Mean absolute error &
$\displaystyle \frac{1}{n}\sum_{i=1}^{n}\left|\hat{f}(x_i)-f_{\mathrm{SciPy}}(x_i)\right|$ \\
Runtime speedup factor against SciPy &
$\displaystyle \frac{T_{\mathrm{SciPy}}}{T_{\mathrm{AUB\text{-}HTP}}}$ \\
\bottomrule
\end{tabular}
\caption{Error and performance metrics comparing $\hat{f}$ against the SciPy reference $f_{\mathrm{SciPy}}$ over $\{x_i\}_{i=1}^{n}$, with $\varepsilon>0$ guarding against division by zero.}
\label{tab:error-metrics}
\end{table}

The tests show close agreement across most of the parameter space, but they also reveal regions where the SciPy reference exhibits irregular numerical behavior. The strongest agreement occurs away from the narrow region around $\alpha=1$, where changes in the parameterization and the presence of logarithmic terms make density evaluation more sensitive. In this region, visual diagnostics are important because the relative difference may become large when the reference density is close to zero or numerically unstable. Figure~\ref{fig:pdf-error-near-mode} illustrates this behavior near $\alpha=1$ using the maximum relative difference measured around the mode. Away from $\alpha=1$, Figure~\ref{fig:pdf-error-near-mode} shows agreement with the SciPy reference.  Results for the remaining error metrics, parameter ranges, and tail evaluations, together with the complete benchmark data, are available in the Github repository\footnote{https://github.com/AUB-HTP/AUB-HTP/tree/main/benchmarks}. 
\begin{figure}[H]
\centering
\includegraphics[width=0.8\linewidth]{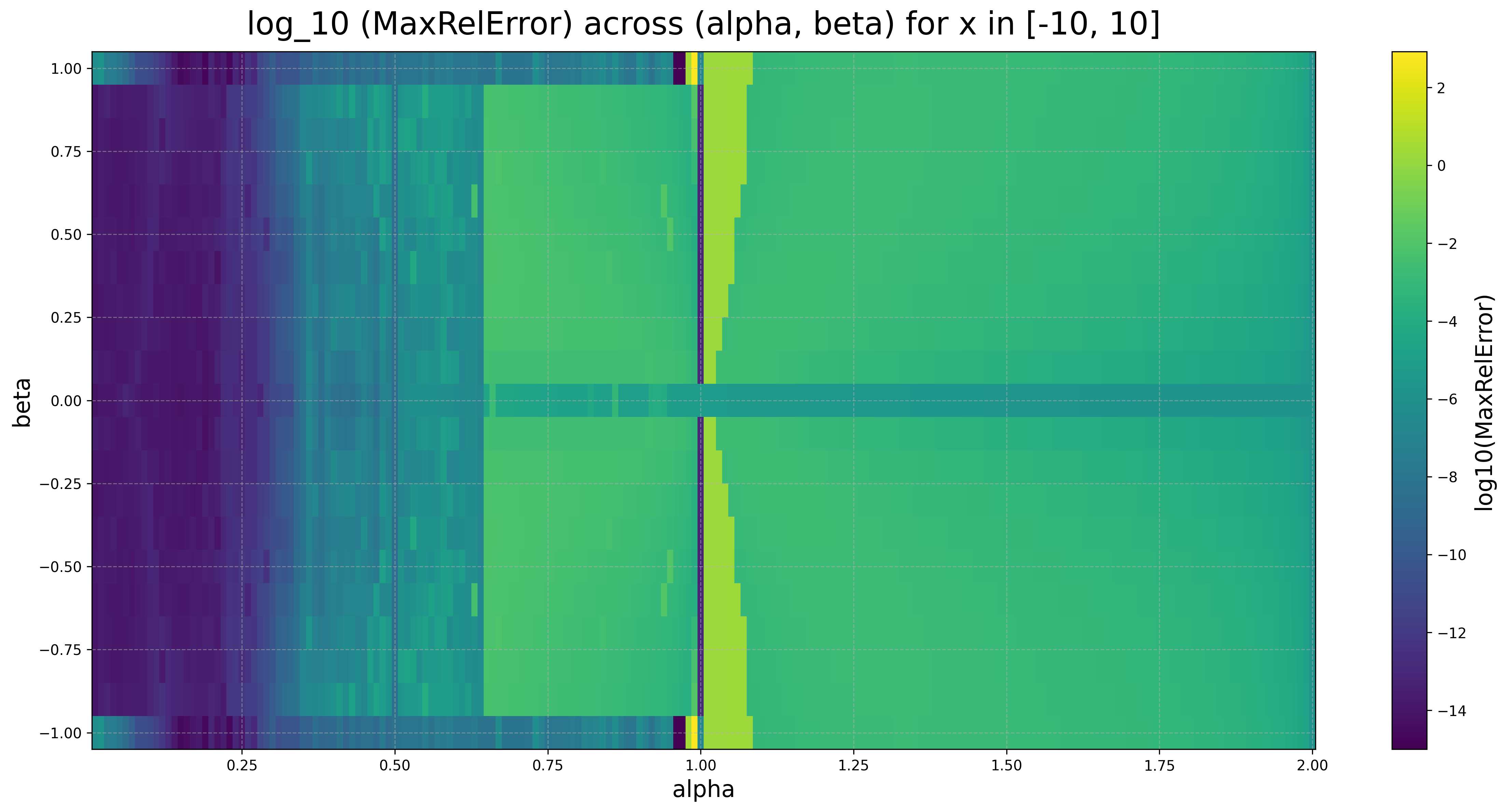}
\caption{Maximum relative difference across $(\alpha,\beta)$ near the mode. Larger values near $\alpha=1$ mark the most sensitive numerical region.}
\label{fig:pdf-error-near-mode}
\end{figure}

\subsection{Discrepancies in the SciPy reference}

While \texttt{scipy.stats.levy\_stable.pdf} \cite{SciPyLevyStableManual} is used as the main numerical reference, we observed that it has some discrepancies for certain parameter values. These discrepancies appear mostly near the sensitive region around $\alpha=1$, and they can create misleading relative difference values.

The main issue is that the SciPy PDF sometimes shows unusual zero or near-zero regions in some places while the densities are known to be smooth. Since the maximum relative difference contains a division by the reference value, even a small disagreement in such regions can produce a very large relative difference. Therefore, some of the large relative differences in the benchmark tables do not necessarily mean that our PDF generator is inaccurate. Instead, they are likely due to the instability or discontinuity in the reference PDF.

To validate this, we generated visual diagnostic plots comparing the AUB-HTP with the \texttt{levy\_stable} PDF values. Figure~\ref{fig:stable-grid-zero-windows} shows several cases where the SciPy reference generates unexpected zero windows or sharp changes, while the PDF generated using our package remains smooth.

\begin{figure}[H]
\centering
\includegraphics[width=\linewidth]{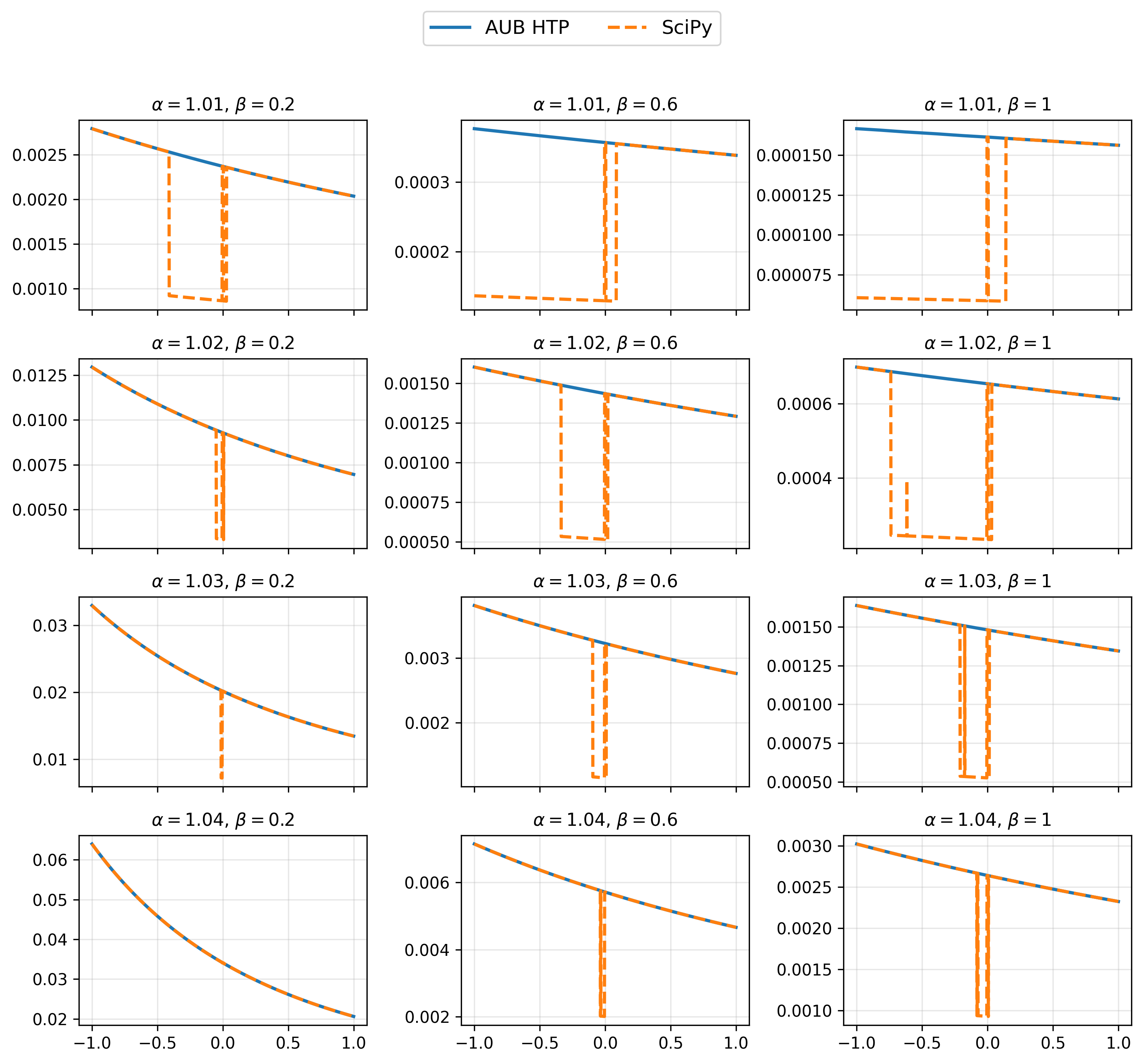}
\caption{Visual comparison between AUB-HTP PDF generator and \texttt{scipy.stats.levy\_stable.pdf}. The plots zoom into regions where the SciPy reference becomes numerically zero or behaves irregularly. These regions explain why some relative difference values become large even when our generated density remains smooth.}
\label{fig:stable-grid-zero-windows}
\end{figure}

The heatmaps are useful for locating sensitive parameter regions, but the diagnostic plots were needed to understand whether the difference is primarily attributed to our method or the numerical reference.

\subsection{Runtime}

The adopted hybrid strategy reduces runtime by avoiding direct quadrature when a faster formula is accurate enough. The largest gains occur in the tail region, where series expansions replace repeated numerical integration.

In the benchmark runs, the median speedup was approximately $58\times$, with larger speedups in tail-heavy evaluations. These gains come from three design choices:

\begin{itemize}
\item vectorized array operations
\item region-specific method dispatch
\item reuse of precomputed cutoff functions.
\end{itemize}

\begin{figure}[H]
\centering
\includegraphics[width=0.9\linewidth]{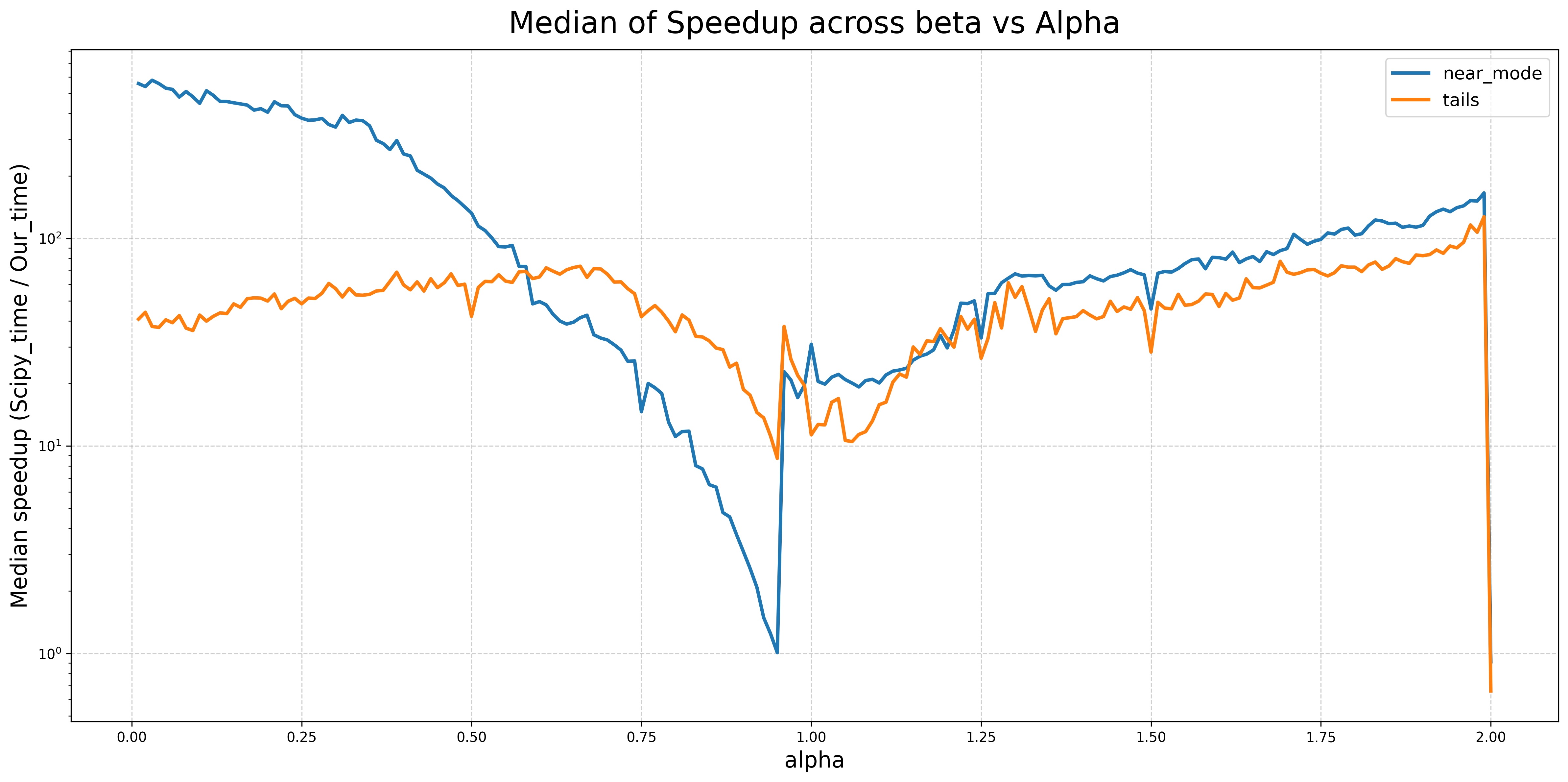}
\caption{Median runtime speedup relative to \texttt{scipy.stats.levy\_stable.pdf}. The largest gains occur where the series expansions are selected.}
\label{fig:pdf-speedup}
\end{figure}

The speedup plot in Figure~\ref{fig:pdf-speedup} shows that our method is strictly faster than \texttt{scipy.stats.levy\_stable.pdf} for almost all tested values of $\alpha$. Note that the vertical axis is on a logarithmic scale and values between $10$ and $100$ represent large runtime improvements. For small values of $\alpha$, especially in the near-mode range, the speedup is very high. This is because the method selection avoids unnecessary numerical integration and uses the formulas that are faster in those regions. In the tail range, the speedup is also consistent for most values of $\alpha$, since the Skorohod-Pollard-Bergstr{\"o}m  formulas replace repeated quadrature calls. The smallest speedups occur near $\alpha=1$ which is expected since this region is numerically delicate. The formulas change around $\alpha=1$, and the characteristic function contains a logarithmic correction term. The package therefore relies more on direct inversion in this region, which is slower than the series-based methods.




\section{Generating $\alpha$-Stable Random Vectors}
In the univariate case, the package implements a numerically stable version of the Chambers, Mallows, and Stuck method which is presented in~\cite{NolanJohnP2020USDM,WERON1996165}, and it is the same method as in \cite{Royuela-del-ValJavier2017l:FP,teimouri2018alphastablerpackagemodelling}. This method is commonly considered to be the fastest and most accurate \cite{Weron2004Computationally}. In the multivariate setup, the proposed random number generator is based on a theoretical foundation that we state next. 

\subsection{The LePage series}

The LePage series~\cite{SamorodnitskyGennady2000Snrp} was first introduced in \cite{LePageRaoul1981CtaS}, and it was suggested as a simulation method in \cite{Weron2004Computationally}. This series allows for simulating a multivariate stable vector with a stability index $\alpha$ and a given spectral measure through the use of the following theorem~\cite[Theorem 4]{BentkusV.2001LSRo}.
\begin{theorem}
\label{thm:spec2}
 Let $\mathbf{V_1},\mathbf{V_2},\cdots$ be independent and identically distributed (i.i.d) random vectors in~$\mathbb{R}^d$ that are drawn according to a spectral measure $\Lambda$ where $\Lambda(\mathbb{S}^{d-1})>0$, i.e.,
 \[  
     \Pr(\mathbf{V}\in A)=\frac{\Lambda(\mathbb{S}^{d-1} \cap A )}{\Lambda(\mathbb{S}^{d-1})} \;\;\forall\;A\in\mathcal{B}(\mathbb{R}^d).
 \] 
 Let $E_1, E_2, \cdots$ be i.i.d. exponential random variables with mean 1, independent of  $\mathbf{V}_1, \mathbf{V}_2, \cdots$, and define
\begin{equation}
\label{LPseries}
\Gamma_k \;=\; \sum_{i=1}^k E_i \qquad
\qquad \mathbf{X}_n= c(\alpha)\sum_{k=1}^{n} \Gamma_k^{-1/\alpha}\,\mathbf{V}_k 
\end{equation}
where  \[
c(\alpha)=\left(\frac{\kappa}{\Lambda(\mathbb{S}^{d-1})}\right)^{-\frac{1}{\alpha}} \qquad \& \qquad \kappa=\begin{cases}
\dfrac{\Gamma(2-\alpha)\,\cos\!\bigl(\tfrac{\pi\alpha}{2}\bigr)}{1-\alpha}
& \alpha\neq1\\[1em]
\dfrac{\pi}{2} 
& \alpha=1.
\end{cases}   . 
\]
\begin{itemize}
\item[$\bullet$] For $0<\alpha<1$, $\mathbf{X}_n$ converges in total variation to a stable vector $\mathbf{X}$ with spectral measure $\Lambda$ and stability index $\alpha$. 
\item[$\bullet$] For $1 \leq \alpha < 2$, whenever
\begin{equation}\label{eq:eq9}
\int_{\mathbb{S}^{d-1}}s\Lambda(ds)=0,
\end{equation}
the random vector $\mathbf{X}_n$ converges in total variation to a stable vector $\mathbf{X}$ with spectral measure $\Lambda$ and stability index $\alpha$.
\end{itemize} 

Whenever there is convergence to $\mathbf{X}$, we write
\[
\mathbf{X}= c(\alpha)\sum_{k=1}^{\infty}\Gamma_k^{-1/\alpha}\,\mathbf{V}_k.
\]
\end{theorem}


Theorem~\ref{thm:spec2} indicates that sampling a multivariate stable vector can be done by computing the truncated series as in equation \eqref{LPseries}. Since the vector $\mathbf{X}$ resulting from the series has a zero shift vector $\boldsymbol{\mu}^0 = {\bf 0}$~\cite[Remark 1]{BentkusV.2001LSRo}, we simply add $\boldsymbol{\mu}^0$ after sampling $\mathbf{X}_n$ using the LePage series, when needed. The analysis of the truncation error is deferred to Appendix~\ref{app:appendixerror}. 

In the following section, we present algorithms to sample according to well-known spectral measures, namely: Isotropic, Sub-Gaussian, Discrete, and Mixed spectral measures.

\subsection{Common spectral measures}

\subsubsection{Discrete spectral measures}

In the case of a discrete spectral measure with $M$ mass points, we store the point masses in an array $[\mathbf{v}_1, \mathbf{v}_2, \cdots, \mathbf{v}_M]$ and draw randomly with the probability of each vector $\mathbf{v}_{i}$ 
being
\[
\Pr(\mathbf{v}_{i})=\frac{\Lambda(\mathbf{v}_{i})}{\sum_{\ell=1}^{M}\Lambda(\mathbf{v}_{\ell})}.    
\]

\subsubsection{Isotropic spectral measures}
\label{sc:Iso}

For an isotropic spectral measure, we sample the vectors $\{\mathbf{V}_k\}_{1 \leq k \leq n}$ uniformly on $\mathbb{S}^{d-1}$. To this end, several algorithms are available and we refer the reader to~\cite{SchnabelStefan2022Asaf} where various such algorithms are discussed. In our computations, we adopt the method presented in~\cite{MullerMervinE.1959Anoa} due to its simplicity and computational efficiency. Our package sets the spectral measure of the sphere to $1$, and then scales the truncated LePage series by an appropriate factor $q$ to adjust for the desired scale $\gamma$ given by the user: When the spectral measure of the sphere is 1, it can be verified that the characteristic function is given by 
\begin{equation*}
    \phi_{\mathbf{X}}(\mathbf{t})=\exp\left\{-\frac{\Gamma\left(\frac{d}{2}\right)}{\sqrt{\pi}}\frac{\Gamma\left(\frac{\alpha+1}{2}\right)}{\Gamma\left(\frac{\alpha+d}{2}\right)}\|\mathbf{t}\|^{\alpha}\right\}.
\end{equation*}

To achieve a scale $\gamma$, we therefore multiply the truncated series by 
\begin{equation*}
    q=\gamma\left[\frac{\Gamma\left(\frac{d}{2}\right)}{\sqrt{\pi}}\frac{\Gamma\left(\frac{\alpha+1}{2}\right)}{\Gamma\left(\frac{\alpha+d}{2}\right)}\right]^{-\frac{1}{\alpha}}.
\end{equation*}

\subsubsection{Sub-Gaussian (elliptical) spectral measures}
\label{sub_guass_error}

A sub-Gaussian random vector is characterized by a  positive-definite $d\times d$ shape matrix $\Sigma$. If $\vect{X}$ is a unit-scale  isotropic multivariate stable random vector with stability index $\alpha$ then $\vect{Y}=\Sigma^{\frac{1}{2}}\vect{X}$ is a sub-Gaussian random vector with shape matrix $\Sigma$ where $\Sigma^{\frac{1}{2}}$ is the Cholesky decomposition of $\Sigma$~\cite{nolancirc2013}.

One way to sample sub-Gaussian random vectors is by using the definition presented in~\cite[Definition 2.5.1]{SamorodnitskyGennady2000Snrp}. This is adopted whenever the $\alpha$-stable vector is purely sub-Gaussian. Whenever the sub-Gaussian spectral measure is an element of a mixture (see section \ref{mixtures} for mixed spectral measures), then the approach in~\cite[Definition 2.5.1]{SamorodnitskyGennady2000Snrp} is no longer suitable.  Instead, we devise a new approach that is based on the following result: if $\mathbf{X}_n$ is the $n$-terms LePage series of an isotropic vector $\mathbf{X}$ with unit scale, then, by Theorem~\ref{thm:spec2}, $\mathbf{X}_n$ converges in total variation to $\mathbf{X}$, and \begin{equation}
\label{required V}
\mathbf{Y}_n=\Sigma^{\frac{1}{2}}\mathbf{X}_n=c(\alpha)\Sigma^{\frac{1}{2}}\sum_{k=1}^n\Gamma_k^{-\frac{1}{\alpha}}\mathbf{U}_k=c(\alpha)\sum_{k=1}^n\Gamma_k^{-\frac{1}{\alpha}}\mathbf{V}_k \qquad \text{where } \mathbf{V}_k=\Sigma^{\frac{1}{2}}\mathbf{U}_k,
\end{equation}
converges in total variation to $\mathbf{Y}=\Sigma^{\frac{1}{2}}\mathbf{X}$, which is the desired sub-Gaussian random vector. Hence, to sample the vectors $\mathbf{V}_k$ for the sub-Gaussian case, the software samples  $\mathbf{U}_k$ uniformly distributed on the unit sphere and multiplies them by $\Sigma^{\frac{1}{2}}$. 


\subsubsection{Mixed spectral measures}
\label{mixtures}
Now consider the scenario where the spectral measure at hand is a mixture of $L$ different spectral measures: $\Lambda_1, \Lambda_2, \cdots, \Lambda_L$. If $w_1, w_2, \cdots, w_L$ are their respective associated weights, the mixed spectral measure $\Lambda_{\text{mix}}$ is given by 
\[
    \Lambda_{\text{mix}}=\sum_{i=1}^L w_i\Lambda_i.
\]

To sample from this spectral measure, we proceed as follows. 
\begin{itemize}
\item[(i)] For each weight $w_i$ we associate its corresponding probability as \[
    p_i = \frac{w_i}{\sum_{\ell=1}^L w_\ell} \qquad i=1, \cdots, L.
\]
Then, we sort the probabilities in non-decreasing order that we assume --without loss of generality-- to be $p_1, \cdots, p_L$ hereafter.

\item[(ii)] We draw   $U\sim\mathcal{U}[0,1]$, 
\begin{itemize}
\item[$\bullet$] If $U\leq p_1$ sample from $\Lambda_1$, 
\item[$\bullet$] else if $U\leq p_1+p_2$, sample from $\Lambda_2$, 
\item[$\bullet$] else iterate on the remaining spectral measures. 
\end{itemize}
\end{itemize}

This algorithm achieves a sampling from the mixture of $\Lambda_1, \cdots, \Lambda_L$ with probabilities $p_1, \cdots, p_L$ respectively.


\section{Implementation}

Whenever the LePage series is used, the package chooses the number of terms to achieve an MSE less than 0.01 using Theorem~\ref{MSE} in Appendix~\ref{app:appendixerror}. This number is capped at 50,000 terms and a warning is issued whenever this cap is reached. After varying the values of $\alpha$ and the total mass of the unit sphere, a suitable value for the mean square error threshold, one that leads to sufficiently accurate results with reasonable consumption of time and memory is 0.01.
\subsection{Univariate random number generator}
In the univariate case, the software samples $X \sim\mathcal{S}(\alpha,\beta,\gamma,\delta;k)$, and the type 1 parameterization is used by default. A sample call for both parametrizations is shown in the following block of code.
\begin{minted}[fontsize=\small, frame=single, linenos]{python}
import aub_htp as ht

alpha = 1
beta = 0
loc = 0
scale = 1
n = 500000
samples = ht.alpha_stable.rvs(alpha=alpha, beta=beta, loc=loc, scale=scale, size=n)
# For the type 0 parametrization use:
# ht.alpha_stable.with_parameterization("S0")
# samples = ht.alpha_stable.rvs(alpha=alpha, beta=beta, loc=loc, scale=scale, size=n)
\end{minted}

\subsection{Multivariate random number generator}

\paragraph{Isotropic $\alpha$-stable random vector:} An isotropic $\alpha$-stable random vector with characteristic function 
\[
    \phi_{\mathbf{X}}(\mathbf{t})=e^{-\gamma^{\alpha}\|\mathbf{t}\|^{\alpha}
+\ii\mathbf{t}^T\boldsymbol{\mu}^0},
\] 
can be generated using the AUB-HTP. A sample syntax --commented out-- is shown below in the section of the Ellipcitc $\alpha$-stable random vector.
\paragraph{Elliptic $\alpha$-stable random vector:} An elliptic $\alpha$-stable random vector with shape matrix $\Sigma$ with the following characteristic function 
\[
    \phi_{\mathbf{X}}(\mathbf{t})=e^{-(\mathbf{t}^T\Sigma \mathbf{t})^{\frac{\alpha}{2}} +\ii\mathbf{t}^T\boldsymbol{\mu}^0},
\]
can also be generated using AUB-HTP. We note that the elliptic sampler takes the total mass of the unit sphere as input. If the total mass of the sphere is not specified, the software estimates it using Monte Carlo simulations (with $10^6$ samples) as follows:
\begin{equation}
    \label{MC_simulation_mass_estimation}
        \Lambda(\mathbb{S}^{d-1})=\mathbb{E}\left[\|\Sigma^{\frac{1}{2}}\mathbf{U}\|^{\alpha}\right]\approx\frac{1}{n}\sum_{i=1}^n\|\Sigma^{\frac{1}{2}}\mathbf{U}_i\|^{\alpha},
    \end{equation}
where the equality is due to~ \cite[proposition 2.5.8]{SamorodnitskyGennady2000Snrp}, and where  $\mathbf{U}_1,\mathbf{U}_2,\ldots\mathbf{U}_n$ are i.i.d. samples according to $\mathbf{U}\sim\mathcal{U}(\mathbb{S}^{d-1})$. Equation~\eqref{MC_simulation_mass_estimation} is justified by the application of the Weak Law of Large Numbers (WLLN) because the random variable $\|\Sigma^{\frac{1}{2}}\mathbf{U}\|^{\alpha}$ has a finite second moment. Indeed, 
\begin{align}
        \mathbb{E}[\|\Sigma^{\frac{1}{2}}\mathbf{U}\|^{2\alpha}]&=\mathbb{E}\left[\left(\mathbf{U}^T\Sigma^{\frac{1}{2}T}\Sigma^{\frac{1}{2}}\mathbf{U}\right)^{\alpha}\right]\nonumber\\&=\mathbb{E}\left[\left(\mathbf{U}^T\Sigma^T\mathbf{U}\right)^{\alpha}\right]\label{cholesky}\\
&=\mathbb{E}\left[\|\mathbf{U}\|^{2\alpha}_{\Sigma}\right]\label{number}\\&\leq \lambda_{\text{max}}^{\alpha}(\Sigma)\mathbb{E}\left[\|\mathbf{U}\|^{2\alpha}\right]=\lambda_{\text{max}}^{\alpha}(\Sigma)<\infty,\label{eq:ineq}
\end{align}
where equation~\eqref{cholesky} is due to the Cholesky decomposition definition and where in equation~\eqref{number} we used the $\Sigma$-weighted norm $\|\mathbf{x}\|^2_{\Sigma}\triangleq\mathbf{x}^T\Sigma\mathbf{x}$. Finally, the inequality in equation~\eqref{eq:ineq} is due to the fact that $\|\mathbf{x}\|^2_{\Sigma}\leq\lambda_{\max}(\Sigma)\|\mathbf{x}\|^2$ where $\lambda_{\text{max}}(\Sigma)$ is the largest eigenvalue of $\Sigma$, equal to the norm of matrix $\Sigma$~\cite{alma991015644117709436}.

Finally, we note that the MSE resulting from estimating the total mass of the sphere using equation~\eqref{MC_simulation_mass_estimation} is given by~\cite[Theorem 1.1]{NiederreiterHarald1992Rnga}.

    An example code for both isotropic and elliptic cases is shown as follows. \begin{minted}[fontsize=\small,breaklines=true, frame=single, linenos]{python}
import aub_htp as ht
from aub_htp.random import IsotropicSampler, EllipticSampler

alpha = 0.5
gamma = 1
shift = [0, 0]
d = 2
n = 500000
sigma1 = [[2, 0.8], [0.8, 1.5]]
sampler = EllipticSampler(number_of_dimensions=2, alpha=alpha, sigma=sigma1)
samples = ht.multivariate_alpha_stable.rvs(alpha=alpha, spectral_measure_sampler=sampler, shift=shift, size=n)
# For an Isotropic spectral measure use:
# sampler = IsotropicSampler(number_of_dimensions=d, alpha=alpha, gamma=gamma)
# samples = ht.multivariate_alpha_stable.rvs(alpha=alpha, spectral_measure_sampler=sampler, shift=shift, size=n)
\end{minted}
Figures \ref{elliptic1} and \ref{elliptic2} show the scatter plots of the elliptical $\alpha$-stable random vector, where we use the following shape matrices
\[
    \Sigma_1=\begin{pmatrix}
2 & 0.8 \\
0.8 & 1.5
\end{pmatrix} \qquad \Sigma_2=\begin{pmatrix}
2 & -1 \\
-1 & 2
\end{pmatrix}.
\]
Figure \ref{isotropic_one} shows the scatterplot of an isotropical vector associated with the commented code.
\begin{figure}[H]
    \centering
    \begin{subfigure}{0.32\textwidth}
        \centering
        \includegraphics[width=\linewidth]{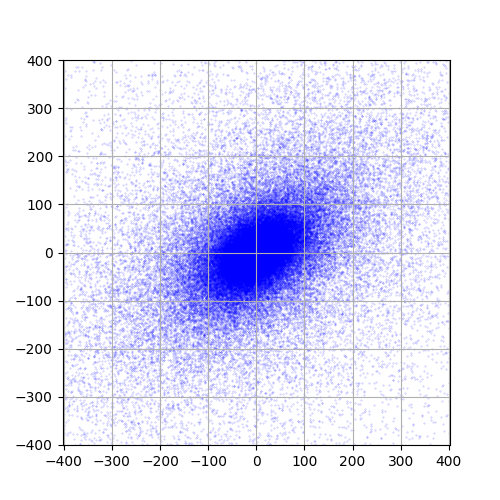}
        \caption{$\alpha=0.5, \Sigma_1, \boldsymbol{\mu}^{0 T}=(0,0)$}
        \label{elliptic1}
    \end{subfigure}
    \hfill
    \begin{subfigure}{0.32\textwidth}
        \centering
        \includegraphics[width=\linewidth]{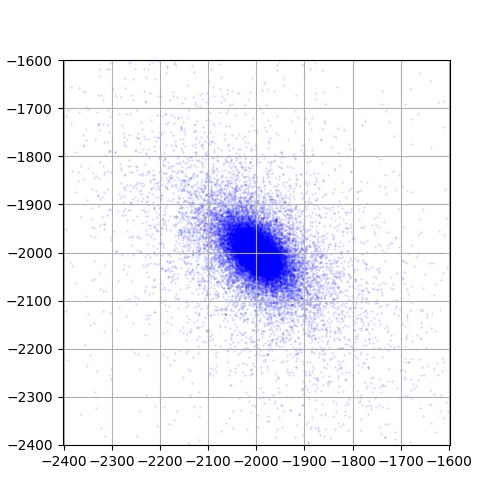}
        \caption{$\alpha=1.1, \Sigma_2, \boldsymbol{\mu}^{0T}=-(2000,2000)$}
        \label{elliptic2}
    \end{subfigure}
    \hfill
    \begin{subfigure}{0.32\textwidth}
        \centering
        \includegraphics[width=\linewidth]{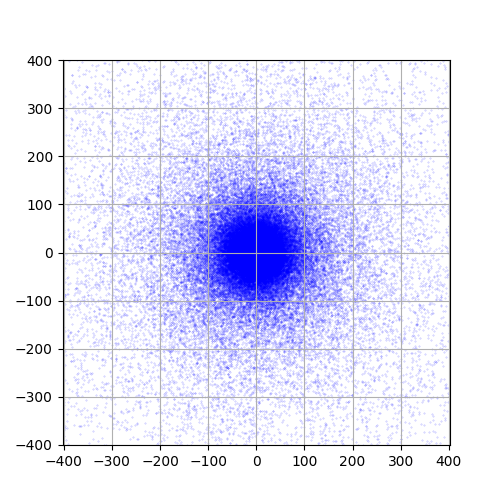}
    \caption{$\alpha=0.5,\gamma=1,\boldsymbol{\mu}^{0 T}=(0,0)$}
    \label{isotropic_one}
    \end{subfigure}
    \caption{Scatter plots each containing $5 \times 10^5$ points of elliptic contoured and isotropic $\alpha$-stable distributions; the plots have been truncated to a range of 400 points within the shift vector.}
    \label{elliptic}
\end{figure}

\paragraph{Discrete spectral measures:} For an $\alpha$-stable random vector that has a discrete spectral measure with point masses at positions $\vect{v}_1, \cdots, \mathbf{v}_M$ and corresponding weights $w_1, \cdots, w_M$, the corresponding characteristic function is given by 
\[
    \phi_{\mathbf{X}}(\mathbf{t})=\exp\left\{-\sum_{\ell=1}^M w_\ell\psi_{\alpha}(\mathbf{t}^T\mathbf{v}_\ell)+\ii\mathbf{t}^T\boldsymbol{\mu}^0\right\}. 
\]  
Such a vector can be generated as follows:
    
\begin{minted}[fontsize=\small,breaklines=true, frame=single, linenos]{python}
import aub_htp as ht
from aub_htp.random import DiscreteSampler

alpha = 0.2
shift = [0, 0]
d = 2
n = 500000
positions = [[1, 0], [0, 1], [-1, 0], [0, -1]]
weights = [1, 1, 1, 1]
sampler = DiscreteSampler(alpha=alpha, positions=positions, weights=weights)
samples = ht.multivariate_alpha_stable.rvs(alpha=alpha, spectral_measure_sampler=sampler, shift=shift, size=n)
\end{minted}
with corresponding scatter plots shown in Figure~\ref{fig:discrete}.

\begin{figure}[H]
    \centering
    \begin{subfigure}[t]{0.32\textwidth}
        \centering
        \includegraphics[width=\linewidth]{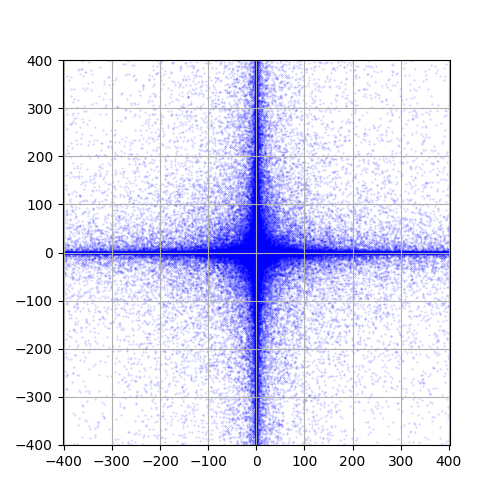}
        \caption{$\text{Masses}=[1,1,1,1], \text{weights}=[(1,0),(0,1),(-1,0),(0,-1)], \alpha=0.2$, $\boldsymbol{\mu}^{0 T}=(0,0)$.}
    \end{subfigure}
    \hfill
    \begin{subfigure}[t]{0.32\textwidth}
        \centering
        \includegraphics[width=\linewidth]{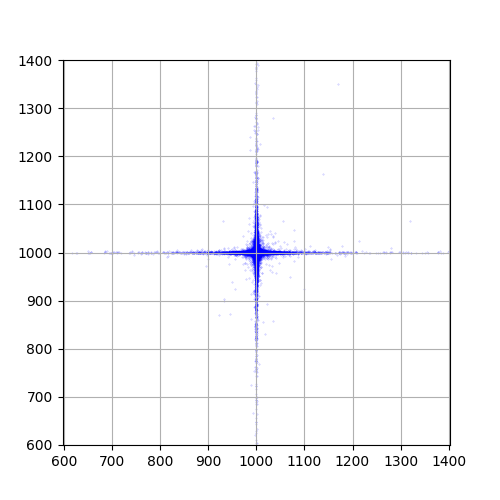}
        \caption{$\text{Masses}=[\frac{1}{4},\frac{1}{4},\frac{1}{4},\frac{1}{4}], \text{weights}=[(1,0),(0,1),(-1,0),(0,-1)], \alpha = 1.2$, $ \boldsymbol{\mu}^{0 T}=(1000,1000)$.}
    \end{subfigure}
    \hfill
    \begin{subfigure}[t]{0.32\textwidth}
        \centering
        \includegraphics[width=\linewidth]{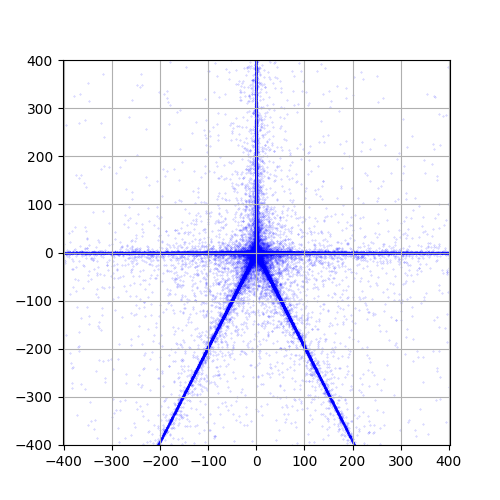}
    \caption{
Masses $=[\frac{1}{5},\frac{1}{5},\frac{1}{5},\frac{1}{5},\frac{1}{5}]$, 
weights $=[(1,0),(0,1),(-1,0),(\cos(\theta),-\sin(\theta)),$ \\ $(\cos(\pi - \theta),-\sin(\pi - \theta))]$, $\theta = \frac{13\pi}{20}$, $\alpha=0.1$,  $\boldsymbol{\mu}^{0T}=(0,0)$.
}
    \end{subfigure}
    \caption{Scatter plots each containing $5 \times 10^5$ points of  $\alpha$-stable distributions with a discrete spectral measure. The plots have been truncated to a range of 400 points within the shift vector. 
    \label{fig:discrete}}
\end{figure}

\paragraph{Mixed spectral measure:} For a mixed spectral measure, let $\Lambda_1, \cdots, \Lambda_L$ be $L$ spectral measures with associated weights $w_1, \cdots, w_L$, and consider their mixture $\Lambda_{\text{mix}} = \sum_{i=1}^L w_i \Lambda_i$. To sample according to this spectral measure, we proceed as follows:

\begin{minted}[fontsize=\small,breaklines=true, frame=single, linenos]{python}
import aub_htp as ht
from aub_htp.random import DiscreteSampler, IsotropicSampler, MixedSampler

alpha = 0.3
shift = [1000, 1000]
d = 2
n = 500000
gamma = 1

# Discrete Spectral Measure
positions = [[1, 0], [0, 1], [-1, 0], [0, -1]]
weights = [1, 1, 1, 1]
sampler1 = DiscreteSampler(alpha=alpha, positions=positions, weights=weights)

# Isotropic Spectral Measure
sampler2 = IsotropicSampler(number_of_dimensions=d, alpha=alpha, gamma=gamma)

# Spectral Measures and Their Associated Weights
spectral_measures = [sampler1, sampler2]
measure_weights = [0.5, 0.5]

mixed_sampler = MixedSampler(spectral_measures=spectral_measures,
                             weights=measure_weights)
samples = ht.multivariate_alpha_stable.rvs(
    alpha=alpha, spectral_measure_sampler=mixed_sampler, shift=shift, size=n)

\end{minted}
where \texttt{spectral\_measures} is a list containing the spectral measures that we are mixing, and \texttt{measure\_weights} contains the associated weight of each measure. The resulting scatter plot is shown in Figure~\ref{mix_isotropic_and_discrete}.

\paragraph{Custom spectral measures:} The common spectral measures presented above are built-in into AUB-HTP. If the user wants to sample from an $\alpha$-stable random vector whose spectral measure is not built-in, this is made possible using a systematic procedure outlined in the following example:  
\begin{itemize}

\item Say the user is interested in sampling a multivariate $\alpha$-stable random vector in dimension $d=2$ with $\alpha=0.3, \boldsymbol{\mu}^0=\mathbf{0}$, and whose spectral measure is uniform over the arcs $\left[-\frac{\pi}{4},\frac{\pi}{4}\right]$ and $\left[\frac{3\pi}{4},\frac{5\pi}{4}\right]$, that is, with probability $0.5$, we  draw a point uniformly on the arc $\left[\frac{3\pi}{4},\frac{5\pi}{4}\right]$, and with probability $0.5$, we draw a point uniformly on the arc $\left[-\frac{\pi}{4},\frac{\pi}{4}\right]$. 

\item We assign equal probabilities to both arcs and the spectral measure of each arc is half the total measure of the sphere. Generally, the probability of each region is the ratio of the spectral measure of the region to the total mass of the sphere. 

\item We create a class of the spectral measure sampler that we want by inheriting from the base class \\ \texttt{BaseSpectralMeasureSampler}. In the class that the user is creating, the function which samples the vectors $\mathbf{V}$ in Theorem \ref{thm:spec2} should be provided. Moreover, the dimension of the desired alpha-stable vector and the total mass of the unit sphere should be specified. The created spectral measure sampler is then passed as an argument to the \texttt{rvs()} method as shown in the following code.
\end{itemize}

The steps are presented below. 
\begin{minted}[fontsize=\small,breaklines=true, frame=single, linenos]{python}
import aub_htp as ht
from aub_htp.random import BaseSpectralMeasureSampler

alpha = 0.3
shift = [0, 0]
d = 2
n = 500000


class ButterflySampler(BaseSpectralMeasureSampler):

    def sample(self, number_of_samples: int, random_state=None):
        p = np.random.rand(number_of_samples)
        theta = np.empty(number_of_samples)

        mask = p <= 0.5
        theta[mask] = np.random.uniform(-np.pi / 4, np.pi / 4, size=mask.sum())
        theta[~mask] = np.random.uniform(3 * np.pi / 4,
                                         5 * np.pi / 4,
                                         size=(~mask).sum())

        x = np.cos(theta)
        y = np.sin(theta)

        return np.column_stack((x, y))

    def dimensions(self) -> int:
        return 2

    def mass(self) -> float:
        return 1.0


samples = ht.multivariate_alpha_stable.rvs(
    alpha=alpha,
    spectral_measure_sampler=ButterflySampler(),
    size=n,
    shift=shift)
\end{minted}

We highlight that the mass of the sphere and the dimension should be specified when implementing the sub-class. The resulting scatter plot is shown in Figure \ref{custom_sp}. It is important to emphasize that when $\alpha \geq 1$, the custom spectral measure should satisfy the condition in equation~\eqref{eq:eq9}.

\begin{figure}[t]
    \centering
    \begin{subfigure}[t]{0.32\textwidth}
        \centering
        \includegraphics[width=\linewidth]{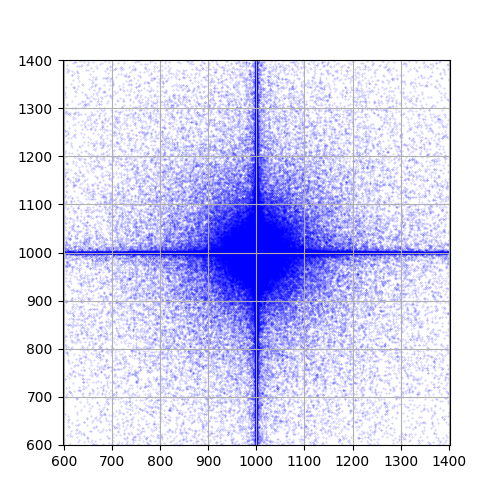}
        \caption{Mix between isotropic and discrete spectral measures.}
        \label{mix_isotropic_and_discrete}
    \end{subfigure}
    \hfill
    \begin{subfigure}[t]{0.32\textwidth}
        \centering
        \includegraphics[width=\linewidth]{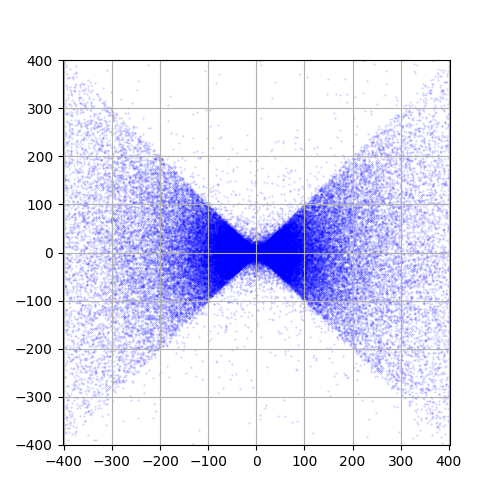}
        \caption{A custom spectral measure}
        \label{custom_sp}
    \end{subfigure}
    \hfill
    \begin{subfigure}[t]{0.32\textwidth}
        \centering
        \includegraphics[width=\linewidth]{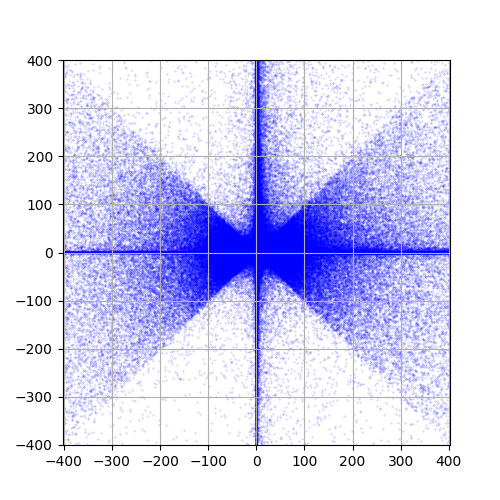}
    \caption{
Mix between custom and discrete spectral measures.}
\label{mix_discrete_and_custom}
    \end{subfigure}
    \caption{Scatter plots showing a custom spectral measure and different mixtures of spectral measures. The plots have been truncated to a range of 400 points within the shift vector. }
    \label{discrete}
\end{figure}


One could mix a custom spectral measure, with another spectral measure. For example, we mix below equally the spectral measure defined above with the discrete spectral measure with $\text{masses}=[(1,0),(0,1),(-1,0),(0,-1)]$ and $ \text{weights}=[1,1,0.25,0.25],$ and the random vector has $\alpha=0.25$, $\boldsymbol{\mu}^{0 T}=(0,0)$.

\begin{minted}[fontsize=\small,breaklines=true, frame=single, linenos]{python}
import aub_htp as ht
from aub_htp.random import BaseSpectralMeasureSampler, DiscreteSampler, MixedSampler

alpha = 0.25
shift = [0, 0]
d = 2
n = 1000000


class ButterflySampler(BaseSpectralMeasureSampler):
    ...
    #implementation in previous block of code.


# Discrete Spectral Measure
positions = [[1, 0], [0, 1], [-1, 0], [0, -1]]
weights = [1, 1, 0.25, 0.25]
sampler1 = DiscreteSampler(alpha=alpha, positions=positions, weights=weights)

#Spectral Measures and Their Associated Weights
spectral_measures = [sampler1, ButterflySampler()]
measure_weights = [0.5, 0.5]

mixed_sampler = MixedSampler(spectral_measures=spectral_measures,
                             weights=measure_weights)
samples = ht.multivariate_alpha_stable.rvs(
    alpha=alpha, spectral_measure_sampler=mixed_sampler, shift=shift, size=n)

\end{minted}

The resulting scatterplot is shown in Figure~\ref{mix_discrete_and_custom}.


\section{Random Variate Generation Validation}

To validate the correctness of our random variate generation, we conducted a projection based goodness of fit test. By the stability property, if $\mathbf{X}$ is an $\alpha$-stable random vector in $\mathbb{R}^d$, then $\boldsymbol{\theta}^T\mathbf{X}$ is an $\alpha$-stable random variable $\forall \boldsymbol{\theta}\in\mathbb{R}^d$ \cite[Theorem 2.1.2]{SamorodnitskyGennady2000Snrp}. One way to test the correctness of our generator is therefore to sample random vectors from AUB-HTP, compute the projection of these vectors on some other vectors, and then test the correctness of the resulting projection values which should theoretically have a univariate $\alpha$-stable distribution. To test the goodness of fit of the resulting projection values, we use the Kolmogorov-Smirnov (KS) test. For details about the formulation of the test, we refer the reader to \cite[Section 10.6]{degroot2018probability}. We use the implementation of the test provided in the STABLE 5.3 Matlab package \cite{robustanalysis}. The test returns the $D$ statistic (see \cite[equation 10.6.2]{degroot2018probability}) and the $p$-value associated with the test. Smaller values of $D$ and larger values of $p$ indicate a good fit. In all of the following tests, we determined the theoretical values of $\beta,\gamma,$ and $\delta$ using \cite[equation. 2.3.3, equation. 2.3.4,  equation. 2.3.5]{SamorodnitskyGennady2000Snrp}; we sampled using AUB-HTP vectors with a mean square error less than 0.001; the dimension of the sampled vector was chosen to be $d=2$; the shift vector was chosen to be $\boldsymbol{\mu}^{0T}=\mathbf{0}$; and 200000 sampled vectors were generated for the test.  We tested two spectral measures, Isotropic and Discrete.
\subsection{Isotropic Test}
 The sampled vectors used in this test have unit scale. Since the spectral measure is isotropic, the distribution of the projected values is invariant to the direction of projection. We therefore project along the $x$-axis and vary only the norm of the projection vector, using $(1,0)$, $(2,0)$, $(3,0)$, $(4,0)$, and $(5,0)$. The resulting one-dimensional projections retain the same stability index $\alpha$ as the original vector, with $\beta=0$, $\delta=0$, and scale parameters $\gamma=1,2,3,4,5$, respectively.

We report the results of the conducted tests for $\alpha=0.6$ and $\alpha=1.1$: the results on all the projected vectors gave $D=0.0021$, $p=0.344$ for  $\alpha=0.6$, and $D=0.0013$, $p=0.9$, for $\alpha=1.1$. Both results indicate good fit.

\subsection{Discrete Test}
In this example, the spectral measure has 4 point masses at $(1,0),(0,1),(-1,0),(0,-1)$ with all masses having weight $0.25$. We project along the vector $(\cos\theta,\sin\theta)$ and we vary $\theta\in\left\{0,\frac{\pi}{4},\frac{\pi}{2},\frac{3\pi}{4},\pi,\frac{5\pi}{4},\frac{3\pi}{2},\frac{7\pi}{4}\right\}.$ Tables \ref{tab:discrete_4mass0.6} and \ref{tab:discrete_4mass1.1} show the theoretical values of $\gamma$ and $\beta$ that the univariate stable projections should have (all projections have $\delta=0$), the $D$ statistic, and the $p$-value. Indeed, we can see that for both values of $\alpha$ and for most projection directions, the $D$ statistic is small and the $p$-value is relatively large, which indicates a good fit. There are a few cases in which the $D$ statistic is small but the $p$-value is small as well. In these cases, albeit the value of $p$ is low, the fit is still deemed satisfactory (by virtue of the low value of $D$). This is due to the known fact that the KS test is highly sensitive at large sample sizes resulting in possibly low $p$-values even when the observed discrepancy is practically negligible. For a more detailed discussion on the sensitivity of the KS test, we refer the reader to \cite[Section 10.6]{degroot2018probability}.
\begin{table}[h]
    \centering
    \begin{minipage}{0.48\textwidth}
        \centering
        \caption{KS Test Results for $\alpha=0.6$}
        \label{tab:discrete_4mass0.6}
        \begin{tabular}{ccccc}
            \toprule
            $\theta$ & $\gamma$ & $\beta$ & $D$ & $p$ \\
            \midrule
            $0$ & $0.314980$ & $0$ & $0.001130$ & $0.960380$ \\
            $\frac{\pi}{4}$ & $0.707107$ & $0$ & $0.001058$ & $0.978556$ \\
            $\frac{\pi}{2}$ & $0.314980$ & $0$ & $0.001915$ & $0.455143$ \\
            $\frac{3\pi}{4}$ & $0.707107$ & $0$ & $0.001633$ & $0.660202$ \\
            $\pi$ & $0.314980$ & $0$ & $0.001130$ & $0.960380$ \\
            $\frac{5\pi}{4}$ & $0.707107$ & $0$ & $0.001058$ & $0.978556$ \\
            $\frac{3\pi}{2}$ & $0.314980$ & $0$ & $0.001915$ & $0.455143$ \\
            $\frac{7\pi}{4}$ & $0.707107$ & $0$ & $0.001633$ & $0.660202$ \\
            \bottomrule
        \end{tabular}
    \end{minipage}
    \hfill
    \begin{minipage}{0.48\textwidth}
        \centering
        \caption{KS Test Results for $\alpha=1.1$}
        \label{tab:discrete_4mass1.1}
        \begin{tabular}{ccccc}
            \toprule
            $\theta$ & $\gamma$ & $\beta$ & $D$ & $p$ \\
            \midrule
            $0$ & $0.532521$ & $0$ & $0.001474$ & $0.777075$ \\
            $\frac{\pi}{4}$ & $0.707107$ & $0$ & $0.001979$ & $0.413518$ \\
            $\frac{\pi}{2}$ & $0.532521$ & $0$ & $0.002252$ & $0.262298$ \\
            $\frac{3\pi}{4}$ & $0.707107$ & $0$ & $0.001771$ & $0.556735$ \\
            $\pi$ & $0.532521$ & $0$ & $0.001474$ & $0.777075$ \\
            $\frac{5\pi}{4}$ & $0.707107$ & $0$ & $0.001979$ & $0.413518$ \\
            $\frac{3\pi}{2}$ & $0.532521$ & $0$ & $0.002252$ & $0.262298$ \\
            $\frac{7\pi}{4}$ & $0.707107$ & $0$ & $0.001771$ & $0.556735$ \\
            \bottomrule
        \end{tabular}
    \end{minipage}
\end{table}

\section{Conclusion}
This paper presented AUB-HTP, a Python package for density computation and random variate generation of $\alpha$-stable distributions. The package is motivated by the practical difficulty of working with $\alpha$-stable laws, whose densities generally lack closed-form expressions and whose multivariate simulation requires careful handling of the underlying spectral measure.

AUB-HTP provides scalar density evaluation through a hybrid strategy that combines input standardization, Zolotarev-type integral representations, series formulas, and numerical inversion of characteristic functions. This design improves robustness across parameter regimes and enables systematic comparisons in terms of accuracy and runtime. The validation results show that the package provides reliable scalar density values and identify parameter regimes in which SciPy's reference implementation appears to suffer from numerical inaccuracies.

The package also implements scalar and multivariate random variate generation. For multivariate $\alpha$-stable vectors, AUB-HTP uses LePage series representations and supports several spectral-measure models, including discrete, isotropic, sub-Gaussian/elliptical, mixed, and custom spectral measures. These features extend the range of $\alpha$-stable models that can be simulated within the Python ecosystem. The truncation-error analysis further provides guidance on the approximation induced by finite LePage-series representations.

By bringing together density evaluation, flexible simulation, validation, runtime analysis, and extensible implementation tools, AUB-HTP aims to make $\alpha$-stable modeling more accessible for computational work involving heavy-tailed data. Future developments may include multivariate PDF computation, CDF evaluation, parameter-estimation routines, further numerical acceleration, and application-oriented modules for communications, signal processing, and machine learning.

\appendix
\section{Error analysis for the LePage series estimate}
\label{app:appendixerror}
In this appendix, we derive an upper bound on the mean square error of the LePage series estimate; we frequently use the following asymptotics of the gamma function from \cite{NIST:DLMF}: Let $x,a,b\in\mathbb{R}$, then we have 
\[
\lim_{x \to \infty} \frac{\Gamma(x+a)}{\Gamma(x+b) \, x^{a-b}} = 1,
\]
then $\forall a,b \in\mathbb{R}$, there exists an $n_0$ such that \begin{equation}\label{eq:eq12}
    \frac{\Gamma(k+a)}{\Gamma(k+b)} <2{k^{a-b}} \qquad \forall k\geq n_0.
\end{equation}
\subsection{The LePage series truncation error}
\begin{theorem}
\label{MSE}
     Let $\mathbf{X}$ and $\mathbf{X}_n$ be as in Theorem~\ref{thm:spec2}, and assume in addition that $\|\mathbf{V}_k\|=1$ for all $k = 1, \cdots, n$. Let $n_0$ be the same as in~\eqref{eq:eq12} then $\forall n\geq \max\{n_0,\frac{2}{\alpha}\}$, $\mathbb{E}\|\mathbf{X}_n-\mathbf{X}\|^2\leq\delta_n$ where  \[
\delta_n=\begin{cases}
 \frac{2\alpha^2}{(1-\alpha)^2}c^{2}(\alpha)n^{2-\frac{2}{\alpha}} &\alpha< 1\\[10pt]
\frac{2\alpha }{2-\alpha}c^2 (\alpha)n^{1-\frac{2}{\alpha}} & \alpha\geq 1.
\end{cases}
\]    
\end{theorem}
\begin{proof}
We start with the case $\alpha<1$. We have 
\[
    \|\mathbf{X}_n-\mathbf{X}\|=\left\|c(\alpha) \sum_{k=n+1}^{\infty}\Gamma_{k}^{-\frac{1}{\alpha}}\mathbf{V}_k\right\|\leq c(\alpha) \sum_{k=n+1}^{\infty}\left\|\Gamma_{k}^{-\frac{1}{\alpha}}\mathbf{V}_k\right\| = c(\alpha)\sum_{k=n+1}^{\infty}\Gamma_k^{-\frac{1}{\alpha}},
\]
which implies
\begin{align}
\mathbb{E} \left[ \|\mathbf{X}_n-\mathbf{X}\|^2 \right] & \leq \mathbb{E} \left[ \left(c(\alpha)\sum_{k=n+1}^{\infty}\Gamma_k^{-\frac{1}{\alpha}}\right)^2 \right]
= c^{2}(\alpha)\left(\sum_{k=n+1}^{\infty} \mathbb{E}\left[\Gamma_k^{-\frac{2}{\alpha}}\right]+2\sum_{n<i<j}\mathbb{E}\left[\Gamma_{i}^{-\frac{1}{\alpha}}\Gamma_{j}^{-\frac{1}{\alpha}}\right]\right)\nonumber\\
&\leq c^{2}(\alpha)\left( \sum_{k=n+1}^{\infty} \mathbb{E}\left[ \Gamma_k^{-\frac{2}{\alpha}}\right]+2\sum_{n<i<j}\sqrt{\mathbb{E}\left[\Gamma_{i}^{-\frac{2}{\alpha}}\right]\mathbb{E}\left[\Gamma_{j}^{-\frac{2}{\alpha}}\right]}\right)  \label{eq:CS}\\
&=c^{2}(\alpha)\left(\sum_{k=n+1}^{\infty} \mathbb{E}\left[ \Gamma_k^{-\frac{2}{\alpha}}\right]+\left(\sum_{k=n+1}^{\infty}\sqrt{\mathbb{E}\left[\Gamma_k^{-\frac{2}{\alpha}}\right]}\right)^2-\sum_{k=n+1}^{\infty}\mathbb{E}\left[\Gamma_k^{-\frac{2}{\alpha}}\right]\right)\label{eq:star}\\
&=c^{2}(\alpha)\left(\sum_{k=n+1}^{\infty}\sqrt{\frac{\Gamma(k-\frac{2}{\alpha})}{\Gamma(k)}}\right)^2 \label{eq:miss}\\
&<c^{2}(\alpha) \left(\sum_{k=n+1}^{\infty}\sqrt{2k^{-\frac{2}{\alpha}}}\right)^2 \label{eq:justgamma}\\
&=2c^{2}(\alpha)\left(\sum_{k=n+1}^{\infty}k^{-\frac{1}{\alpha}}\right)^2\nonumber\\
&\leq \frac{2\alpha^2}{(1-\alpha)^2}c^{2}(\alpha)n^{2-\frac{2}{\alpha}} \label{integral_bound}
\end{align}
where equation~\eqref{eq:CS} is due to the Cauchy-Schwarz inequality. Equation~\eqref{eq:miss} is due to the fact that $\Gamma_k\sim\text{Gamma}(k,1),n\geq \frac{2}{\alpha}$, and $\mathbb{E}\left[\Gamma_k^{-\frac{2}{\alpha}}\right]=\frac{\Gamma(k-\frac{2}{\alpha})}{\Gamma(k)}$, and inequality~\eqref{eq:justgamma} is due to~\eqref{eq:eq12}. Inequality \eqref{integral_bound} is due to the fact that \begin{equation}
    \sum_{k=n+1}^{\infty}k^{-\frac{1}{\alpha}}\leq\int_{n}^{\infty}x^{-\frac{1}{\alpha}}dx=\frac{\alpha}{1-\alpha}n^{1-\frac{1}{\alpha}}.
    \end{equation} 
    

For the case $\alpha\geq1$, since $\|\mathbf{V}_k\|=1\;\forall \;k=1,\cdot\cdot\cdot,n$, equation~\eqref{eq:eq9} implies $\mathbb{E}[\mathbf{V}]=\mathbf{0}$. Let $\mathbf{V}_j=(v_j^{1},v_j^{2},\cdot\cdot\cdot,v_j^{d})$ then we have 
\begin{align}
\mathbb{E} \left[ \|\mathbf{X}_n - \mathbf{X}\|^2 \right]
&= c^{2}(\alpha)\,
\mathbb{E}\left[
  \Bigl\langle
    \sum_{\ell=n+1}^{\infty}\Gamma_k^{-\frac{1}{\alpha}}\mathbf{V}_{\ell},
    \sum_{k=n+1}^{\infty}\Gamma_k^{-\frac{1}{\alpha}}\mathbf{V}_k
  \Bigr\rangle
\right]\nonumber\\
&=c^{2}(\alpha)\mathbb{E}\Biggl[\;\sum_{k=n+1}^{\infty}\sum_{\ell=n+1}^{\infty}\sum_{i=1}^{d}\Gamma_k^{-\frac{1}{\alpha}}\Gamma_{\ell}^{-\frac{1}{\alpha}}v_{k}^{i}v_{\ell}^{i}\Biggr]\nonumber\\
&=c^{2}(\alpha)\left[\;\sum_{k=n+1}^{\infty}\sum_{\ell=n+1}^{\infty}\sum_{i=1}^{d}\mathbb{E}\left[\Gamma_k^{-\frac{1}{\alpha}}\Gamma_{\ell}^{-\frac{1}{\alpha}}\right]\mathbb{E}\left[v_{k}^{i}v_{\ell}^{i}]\right]\right] \label{eq:perp}\\
&=c^{2}(\alpha)\sum_{k=n+1}^{\infty}\mathbb{E}\left[\Gamma_k^{-\frac{2}{\alpha}}\right]\mathbb{E}\left[\|\mathbf{V}_k\|^2\right] \label{eq:mean0}\\
&=c^{2}(\alpha)\sum_{k=n+1}^{\infty}\frac{\Gamma(k-\frac{2}{\alpha})}{\Gamma(k)}<2c^{2}(\alpha)\sum_{k=n+1}^{\infty}k^{-\frac{2}{\alpha}}\nonumber\\
&\leq\frac{2\alpha}{2-\alpha}c^{2}(\alpha)n^{1-\frac{2}{\alpha}}, \nonumber
\end{align}
where equation~\eqref{eq:perp} is due to the fact that $\mathbf{V}_k \perp\!\!\!\perp \Gamma_{k}$ and equation~\eqref{eq:mean0} is justified since $v_k^{i} \perp\!\!\!\perp v^{i}_{\ell}$, $\forall \, k\neq\ell$ and $\mathbb{E}[v^{i}_k]=0$, $\forall \,k$. The rest of the equalities and inequalities follow by similar steps as in the case $\alpha<1$. 
\end{proof}

One can clearly see that the bound on the error decays faster for smaller values of $\alpha$. This indicates that this sampling method can appropriately handle the difficult region of small values of $\alpha$. In the one dimensional case, the mean square error of the LePage series can be found in \cite{JanickiAleksander1992Ciot}.


\bibliographystyle{plain}  
\bibliography{ref} 
\end{document}